\documentclass[11pt]{article}

\usepackage{amsmath, amssymb, mathrsfs, amsthm, amsfonts}
\usepackage[margin=3cm]{geometry}
\usepackage[utf8]{inputenc}
\usepackage{times}
\newtheorem{thm}{Theorem}[section]   
\newtheorem{defn}[thm]{Definition}
\newtheorem{lem}[thm]{Lemma}
\newtheorem{ex}{Example}

\newcommand{\fraction}[2]{\raisebox{0.5ex}{$ #1 $}\slash\raisebox{-0.5ex}{$ #2 $}}

\def\NN{\mathbb N}

\def\B{{\mathcal B}}
\def\S{{\mathcal S}}
\def\F{{\mathcal F}}

\title{Use of Signed Permutations in Cryptography}
\author{Iharantsoa Vero RAHARINIRINA\\\texttt{ihvero@yahoo.fr}\\\small{Department of Mathematics and computer science},\\\small{Faculty of Sciences, BP 906 Antananarivo 101, Madagascar}}

\begin{document}
\maketitle
\begin{abstract}
In this paper we consider cryptographic applications of the arithmetic on the hyperoctahedral group. On an appropriate subgroup of the latter, we particularly propose to construct public key cryptosystems based on the discrete logarithm. The fact that the group of signed permutations has rich properties provides fast and  easy implementation and makes these systems resistant to attacks like the Pohlig-Hellman algorithm. The only negative point is that storing and transmitting permutations need large memory. Using together the hyperoctahedral enumeration system and what is called subexceedant functions, we define a one-to-one correspondance between natural numbers and signed permutations with which we label the message units.
\end{abstract}
\section{Introduction}
The term cryptography refers to the study of techniques providing information security. Several mathematical objects have been used in this field to label the so-called message units. These objects are often integers. They may also be points or vectors on some curves. For this work, we shall use signed permutations.\\
Let us denote by:
\begin{itemize}
	\item $ [n] $ the set $ \{1,\cdots ,n\} $,
	\item $ [\pm n] $ the set $ \{-n,\cdots,-1,1,\cdots ,n\} $,
	\item $ \mathcal{S}_n$ the symmetric group of degree $ n $.
\end{itemize}
\begin{defn}
	A bijection $ \pi :[\pm n]\longrightarrow [\pm n] $ satisfying $ \pi(-i)=-\pi(i) $ for all $ i \in [\pm n] $ is called "signed permutation".
\end{defn}
We can also write a signed permutation $ \pi $ in the form
\begin{equation*}
\pi = \left( \begin{array}{cccc} 1 & 2 & \dots & n\\ 
\varepsilon_1 \sigma_1 & \varepsilon_2 \sigma_2 & \dots & \varepsilon_n\sigma_n \end{array} \right) \; 
\text{ with } \sigma\in \mathcal{S}_n \text{ and } \varepsilon_i \in \{\pm 1\}\; .
\end{equation*}
Under the ordinary composition of mappings, all signed permutations of the elements of $ [n] $ form a group $ \B_n $ called hyperoctahedral group of rank $ n $. We write $ \pi^k=\underbrace{\pi\circ \cdots \circ\pi}_{k\text{-times}} $ for an integer $ k $ and $ \pi\in\B_n $ and when we multiply permutations, the leftmost permutation acts first. For example,
$$  
\left( \begin{array}{crcc} 1 & 2 & 3 & 4\\ 
1 & -3  & 4 & 2 \end{array} \right)\circ\left( \begin{array}{crcc} 1 & 2 & 3 & 4\\ 
3 & -2  & 4 & 1 \end{array} \right)=\left( \begin{array}{crcr} 1 & 2 & 3 & 4\\ 
3 & -4  & 1 & -2 \end{array} \right)\; .
$$
Before the 1970's, to cipher or decipher a message, two users of cryptographic system must safely exchange information (the private key) which is not known by anyone else. New keys could be periodically distributed so as to keep the enemy guessing.
In 1976, W. Diffie and M. Hellman \cite{Diffie} discovered an entirely different type of cryptosystem for which all of the necessary information to send an enciphered message is publicly available without enabling anyone to read the secret message. With this kind of system called public key, it is possible for two parties to initiate secret communications without exchanging any preliminary information or ever having had any prior contact.

The security of public key cryptosystems  is based on the hardness of some mathematical problems. As a public key cryptosystem consists of a private key (the deciphering key) that is kept secret and a public key (the enciphering key) which is accessible to the public, then the straightforward way to break the system is to draw the private key from the public key. Therefore, the required computation cost is equivalent to solving these difficult mathematical problems.
For instance for the two particularly important examples of public key cryptosystems : RSA and Diffie-Hellman, both are connected with fundamental questions in number theory which are : difficulty of factoring a large composite integer whose prime factors are not known in advance and intractability of discrete logarithm in the multiplicative group of a large finite field, respectively. 
\begin{defn}
	The discrete logarithm problem in the finite group $ G $ to the base $ g\in G $ is the problem : given $ y\in G $, of finding an integer $ x $ such that $ g^x = y $, provided that such an integer 
	exists (in other words, provided that $ y $ is in the subgroup generated by $ g $). 
\end{defn} 
If we really want our random element $ y $ of $ G $ to have a discrete logarithm, $ g $ must be a generator of $ G $.
\begin{ex}
	Let $ G=(\fraction{\mathbb{Z}}{19\mathbb{Z}})^* $ be the multiplicative group of integers modulo $ 19 $. The successive powers of $ 2 $ reduced$ \mod{19} $ are : $ 2, 4, 8, 16, 13, 7, 14, 9, 18, 17, 15, 11, 3, 6, 12, 5, 10, 1 $. Let $ g $ be the generator $ 2 $, then the discrete logarithm of $ 9 $ to the base $ 2 $ is $ 8 $.
\end{ex}
In many ways, the group of signed permutations is analogous to the multiplicative group of a finite field. 
For this purpose, cryptosystems based on the latter can be translated to systems using the hyperoctahedral group.

This work addresses public key cryptosystems based on the group of signed permutations and related to the discrete logarithm problem. We shall illustrate this in section \ref{system} by an hyperoctahedral group analogue for the Diffie-Hellman key exchange system then by describing two hyperoctahedral group public key cryptosystems for transmitting information. We shall also explain why there is no subexponential time algorithms to break the proposed systems. Before introducing the cryptosystems themselves, we give in section \ref{translate} a method  to define a bijection between integers and signed permutations. This time, we first convert the natural number in hyperoctahedral system and then we use the result to compute the corresponding signed permutation by means of subexceedant function.
\section{Integer representation of signed permutations}\label{translate}
\subsection{Converting from one base to another}
\begin{defn}\label{classical}
	Hyperoctahedral number system is a system that expresses all natural number $ n$ of $ \NN $ in the form : 
	\begin{equation}\label{hyper}
		n=\sum_{i=0}^{k(n)} d_i. B_i \; , \text{ where }k(n)\in\NN,\;  d_i \in \{0,1,2,\cdots,2i+1 \}\; \text{ and }\; B_i =2^i i!\;   .
	\end{equation}
\end{defn}
This definition is motivated by the fact that $ B_i $ is the cardinal of the hyperoctahedral group $ \B _i $. 

To denote the non-negative integer $ n $ of the equation (\ref{hyper}) in the hyperoctahedral system, we use by convention the $ (k+1) $-digits representation
$$  d_k d_{k-1} \cdots d_2 d_1 d_0\; . $$
\begin{thm}\label{one}
	Every positive integer has an unique representation in the hyperoctahedral system.
\end{thm}
Before giving the proof of this theorem, these are some properties of the hyperoctahedral system.

\begin{lem}\label{lem1}
	If $ n= d_k \cdots d_1 d_0 $ is a number in hyperoctahedral system, then
	$$ 0\leqslant n\leqslant B_{k+1}-1 $$
\end{lem}
\begin{proof}
	Since $ 0\leqslant d_i \leqslant 2i+1 $, we have
	$$ 0\leqslant \sum_{i=0}^{k} d_i B_i \leqslant \sum_{i=0}^{k} (2i+1) B_i\, .$$
	Recall that $ B_i =2^i i! $ ,
\begin{eqnarray*}
(2i+1) B_i &=&(2(i+1-1)+1)B_i\\&=&(2(i+1)-1)B_i\\&=& 2(i+1)B_i-B_i\\&=& B_{i+1}-B_i\, .
\end{eqnarray*}
	Therefore,
	$$ \sum_{i=0}^{k} (2i+1) B_i =\sum_{i=0}^{k} (B_{i+1}-B_i) =B_{k+1}-B_0=B_{k+1}-1 \; .$$
\end{proof}
\begin{lem}\label{lem2}
	Let $ n= d_k \cdots d_1 d_0 $ be a number in hyperoctahedral system, then
	$$ d_k B_k\leqslant n < (d_k +1)B_k \; .$$
\end{lem}
\begin{proof}
	Let $ m= d_{k-1} \cdots d_1 d_0 $. From lemma \ref{lem1}, $ 0\leqslant m < B_k  $. We also have $ n=d_k B_k +m $, hence
	\[ d_k B_k\leqslant n < B_k +d_k B_k= (1+d_k )B_k \; . \]
\end{proof}
Now, we can proceed to the demonstration of theorem \ref{one} which states the one-to-oneness between non negative integers and hyperoctahedral base numbers.
\begin{proof}
Let $ a_n \cdots a_ 1 a_0 $ with $ a_n \neq 0 $ and $ b_m \cdots b_1 b_0  ,  b_m \neq 0 $ be two representations of a positive integer $ N $ in hyperoctahedral system.
First, $ b_m \geq 1 $ and $ a_n \geq 1 $ imply $$ B_m\leq b_mB_m\leq b_m \cdots b_1 b_0\;\text{ and }B_n\leq a_n \cdots a_ 1 a_0  . $$The relation $ n=m $ is immediate because if $ n<m $ then 
	\begin{equation*}
	B_m\geq B_{n+1}>a_n \cdots a_ 1 a_0   
	\end{equation*}
	by lemma \ref{lem1} so 
	$$
	b_m \cdots b_1 b_0>a_n \cdots a_ 1 a_0  \; .
	$$
	Similarly, if $ m<n $  then
	$$
	a_n \cdots a_ 1 a_0 > b_m \cdots b_1 b_0\; .
	$$
	We get a contradiction. Consequently $ n=m $ and $ a_i=b_i $ for all $ i\in\{0,1,\ldots,n\} $ by induction and by the unicity of the expression $ N=a_n B_n+r_n $ with $ r_n=a_{n-1} \cdots a_ 1 a_0<B_n $.
\end{proof}
To express a positive integer $n$ in the hyperoctahedral system, one proceeds with the following manner. Start by dividing $n$ by $2$ and let $d_0$  be the rest $r_0 $  of the expression
$$n=r_0+(2)q_0 \; .$$
Divide $q_0$ by $4$, and let $d_1$  be the rest $r_1 $ of the expression
$$q_0=r_1+(4)q_1 \; .$$ 
Continue the procedure by dividing $q_{i-1}$ by $2(i+1) $ and taking $d_i:=r_i$ of the expression
$$q_{i-1}=r_i+2(i +1)q_i \; $$
until $ q_l =0 $ for some $ l\in \NN $. In this way, we obtain $ n= d_l:d_{l-1}: \cdots :d_1 :d_0 $ and we also have
\begin{equation*}
	n= d_0 + 2 \left(d_1 + 4 \cdot(d_2 + 2(3) \cdot(d_3 + \cdots))\right) .
\end{equation*}
Now let	$ n= d_{k-1}:d_{k-2}: \cdots :d_1 :d_0 $ be a number in hyperoctahedral system. By definition \ref{classical}, one way to convert $ n $ to the usual decimal system is to calculate
$$ d_{k-1} 2^{k-1}(k-1)! +\cdots +d_1 .2+d_0 \; . $$
In practice, one can use this algorithm :
\begin{center}
	\begin{tabular}{c}
		\hline\\
		\begin{tabular}{lc}
			\verb|Input :|
			& \verb|An integer| $ d_{k-1}:d_{k-2}: \cdots : d_1 :d_0 $ \verb|in hyperoctahedral system.|\\
			\verb|Output :| & \verb|An integer| $ d $ \verb|in decimal system.|\\
			\hline\\
		\end{tabular}\\
		\begin{tabular}{ll}
			1. \textbf{initiate the value of $ d $ }: &  $ d  \leftarrow d_{k-1} $\\
			2. \textbf{for} $ i $ from $ k-1 $ to 1 do : &  $ d \leftarrow  d.2.i+d_{i-1} $\\
			3. \textbf{return} $ d $& 
		\end{tabular}\\
		\hline
	\end{tabular}
\end{center} \vspace{3mm}
\subsection{A bijection between the subexceedant functions and permutations}
	\begin{defn}
		A subexceedant function $ f $ on $ [n] $ is a map $ f : [n]\longrightarrow[n] $ such that
		$$ 1 \leqslant f (i) \leqslant i \text{ for all } i\in[n]\; .$$
	\end{defn}
	We will denote by $ \F_n $ the set of all subexceedant functions on $ [n] $, and we will represent a subexceedant function $ f $ over $ [n] $ by the word $ f (1) f (2) \cdots f (n) $.
	\begin{ex}
		These are the sets $ \F_n $ for $ n=1, 2, 3 $ :
		\begin{eqnarray*}
			&&\F_1=\{1\}\\
			&&\F_2=\{11,12\}\\
			&&\F_3=\{111,112,113,121,122,123\} .
		\end{eqnarray*}
	\end{ex}
	It is easy to verify that $ card\; \F_n = n! $, since from each subexceedant $ f $ over $ [n-1] $, one can obtain $ n $ distinct subexceedant functions over $ [n] $ by adding any integer $ i \in [n] $ at the end of the word representing $ f $.
	
We will give a bijection between $ \S_n $ and $ \F_n $.
	Let be the map 
	$$
	\begin{array}[pos]{rcl}
	\phi : \F_n&\longrightarrow&\S_n\\
	f& \longmapsto& \sigma_f = (1\, f(1))(2\, f(2))\cdots (n\, f(n))
	\end{array}.
	$$
	Notice that there is an abuse of notation in the definition of $ \sigma_f $ . Indeed, if $ f (i) = i $, then the cycle $ (i\, f(i)) = (i) $ does not really denote a transposition but simply the identity permutation.
	\begin{lem}
		The map $ \phi $ is a bijection from $ \F_n $ onto $ \S_n $.
	\end{lem} 
	\begin{proof}
		Since $ \S_n $ and $ \F_n $ both have cardinality $ n! $, it suffices to prove that $ f $ is injective. Let $ f $ and $ g $ be two subexceedant functions on $ [n] $. Assume that $ \phi( f ) = \phi(g) $ i.e. $  \sigma_f=\sigma_g $. So we have :
		\begin{equation}\label{eq_f=g}
		(1\, f(1))(2\, f(2)) \cdots (n\, f(n)) = (1\; g(1))(2\; g(2)) \cdots (n\; g(n)) .
		\end{equation}
		Since $  \sigma_f=\sigma_g $, then in particular $  \sigma_f(n)=\sigma_g(n) $. By definition $ \sigma_f (n)= f (n) $ and $ \sigma_g(n)=g(n) $, so $ f (n) = g(n) $. Let us multiply both members of equation (\ref{eq_f=g}) on the right by the permutation $ (n\; f(n))=(n\; g(n)) $, we obtain :
		\begin{equation*}
		(1\, f(1))(2\, f(2)) \cdots (n-1\hspace{5mm} f(n-1)) = (1\; g(1))(2\; g(2)) \cdots (n-1\hspace{5mm} g(n-1)) .
		\end{equation*}
		Now, if we apply the same process to this equation, we obtain $ f (n-1) = g(n-1) $. By iterating, we 
		conclude that $ f (i) = g(i) $ for all integers $ i\in[n] $ and then $ f = g $.
	\end{proof}
	Let $ \sigma $ be a permutation of the symmetric group $ \S_n $ and $ f $ be the inverse image of $ \sigma $ by $ \phi $. Then $ f $ can be constructed
	as below :
	\begin{enumerate}
		\item Set $ f (n) = \sigma(n) $.
		\item Multiply $ \sigma $ on the right by $ (n\; \sigma(n)) $ (this operation consists in exchanging the image of $ n $ and the image of $ \sigma^{-1}(n) $), we obtain a new permutation $ \sigma_1 $ having $ n $ as a fixed point. Thus $ \sigma_1 $ can be considered as a permutation of $ \S_{n-1} $. Then set $ f (n-1) = \sigma_1(n-1) $.
		\item In order to obtain $ f (n-2) $, apply now the same process to the permutation $ \sigma_1 $ by multiplying $ \sigma_1 $  by $ (n-1\hspace{5mm} \sigma_1(n-1)) $. Iteration determines $ f (i) $ for all integers $ i $ of $ [n] $.
	\end{enumerate}
\subsection{Mapping hyperoctahedral base numbers to signed permutations}
Let	$ d_{n-1}\; d_{n-2}\;  \cdots\;  d_1\; d_0 $ be a $ n $-digits number in hyperoctahedral system. That is $   d_i \in \{0,1,2,\cdots,2i+1 \} $ for $ i=0, \cdots,n-1 $.
\\	
Writing $ d_i =2q_i+r_i $ where $ r_i\in\{0,1\} $ and $ q_i\in\{0,\ldots,i\} $, gives
	\begin{itemize}
		\item	 the subexceedant function $ f=f(1)\cdots f(n) $ with $ f(i)=1+q_{i-1},\;   i=1,\ldots ,n $
		\item and the sequence $ (\varepsilon_1,\ldots,\varepsilon_n) $ with $ \varepsilon_i=(-1)^{r_{i-1}} $ for $ i=1,\ldots ,n  $.
	\end{itemize}
So to each integer $ d_{n-1}\; d_{n-2}\;  \cdots\;  d_1\;  d_0 $, we associate the signed permutation
	\begin{equation*}
	\left( \begin{array}{cccc} 1 & 2 & \dots & n\\ 
	\varepsilon_1 \sigma_1 & \varepsilon_2 \sigma_2 & \dots & \varepsilon_n\sigma_n \end{array} \right) 
	\end{equation*}
	where $ \sigma= \sigma_1 \cdots \sigma_n $ is the permutation associated to the subexceedant function $ f $ by the map 
	$$
	\begin{array}[pos]{rcl}
	\phi : \F_n&\longrightarrow&\S_n\\
	f& \longmapsto& \sigma_f = (1 f (1))(2 f (2))\cdots (n f (n))
	\end{array}.
	$$
	Now, let $ \pi=\left( \begin{array}{cccc} 1 & 2 & \dots & n\\ 
	\varepsilon_1 \sigma_1 & \varepsilon_2 \sigma_2 & \dots & \varepsilon_n\sigma_n \end{array} \right) 
	$ be a signed permutation.
	As $ \phi $ is a bijection, from $ \pi $ we have
	$$
	f=f(1)\cdots f(n)=\phi^{-1}(\sigma)\text{ and }
	r_{i-1}=\begin{cases}
	0& \text{ if }\varepsilon_i=1\\
	1& \text{ if }\varepsilon_i=-1
	\end{cases}\text{ for }i=1, \ldots,n.
	$$
	The digits $ d_i=2(f(i+1)-1)+r_i,\; i=0,\ldots,n-1 $ form the hyperoctahedral number  	$ d_{n-1}\; d_{n-2}\;  \cdots\;  d_1\;  d_0 $. It is easy to verify that $ d_i\in\{0,\ldots,2i+1\} $. We have
\begin{eqnarray*}
&	1\leqslant f(i)\leqslant i &\text{ for }i=1,\ldots,n\\
&	0\leqslant f(i+1)\leqslant i+1 &\text{ for }i=0,\ldots,n-1\\
&	0\leqslant 2(f(i+1)-1)\leqslant 2i &\text{ for }i=0,\ldots,n-1\\
&	0\leqslant d_i\leqslant 2i+1&\text{ because }0\leqslant r_i\leqslant 1.
\end{eqnarray*}
\section{Some cryptosystems based on hyperoctahedral group}\label{system}
We suppose that we are using message units with signed permutations as equivalents (according to section \ref{translate}) in some publicly known large hyperoctahedral group $ \B_n $.
\subsection{Analog of the Diffie-Helman key exchange}
The Diffie-Hellman key exchange \cite{Diffie} which was the first public key cryptosystem originally used the multiplicative group of a finite field. It can be adapted for signed permutations group as follow. 

Suppose that two users Alice and Bob want to agree upon a secret key, a random element of $ \B_n $, which they will use to encrypt their subsequent messages to one another. They first choose a large integer $ n $ for the hyperoctahedral group $ \B_n $ and select a signed permutation $ \beta\in\B_n $ to serve as their "base" and make them public.
Alice selects a random integer $ 0<a<B_n $, which she keeps secret, and computes $ \beta^a\in\B_n $ which she makes public. Bob does the same. He selects a random integer $ 0<b<B_n $, and transmits $ \beta^b $ to Alice over a public channel. Alice and Bob agree on the secret key $ \beta^{ab} $ by computing $ (\beta^b)^a $ and $ (\beta^a)^b $ respectively.
\subsection{Analog of the ElGamal system}
Suppose user "A" requires sending a message $ \mu $ to user "B". As in the key exchange system above, we start with a fixed publicly known hyperoctahedral group $ \B_n $ for a large integer $ n $ and a signed permutation $ \beta\in\B_n $ (preferably, but not necessarily, a generator) as base. Each user selects a random integer $ 0< u<|\beta| $ (instead of $ u $ we take $ a $ for user "A" and $ b $ for user "B") which is kept secret, and computes $ \beta^u $ that he publishes as a public key.  With the public key $ \beta^b $, user "A" encrypt the message $ \mu $ and sends the pair of signed permutations $ (m_1=\beta^a,m_2=\mu.(\beta ^{b})^a) $ to user "B". To recover $ \mu $, B computes $ (m_1^b)^{-1}=((\beta^a)^b)^{-1} $ and multiplies $ m_2 $ on the right by the result. 
\subsection{Analog of the Massey-Omura system for message transmission.}
We assume the same setup as in the previous subsection. Suppose that Alice wants to send Bob a message $ \mu $.
She chooses a random integer $ c $ satisfying $ 0 < c < B_n $ and $ g.c.d.(c,B_n) = 1 $, and
transmits $ \mu^c $ to Bob. Next, Bob chooses a random integer $ d $ with the same properties, and
transmits $ (\mu^c)^d $ to Alice. Then Alice transmits $ (\mu^{cd})^{c'}=\mu^d $ back to Bob, where $ c'c\equiv 1\mod{B_n} $. Finally, Bob computes $ (\mu^d)^{d'} $ where $ d'd\equiv 1\mod{B_n} $.
\subsection{Security}
For the Diffie-Hellman's key exchange above, a third party knows only the elements $ \beta\ ,\beta^a $ and  $ \beta^b $ of $ \B_n $ which are public knowledge. Obtaining $ \beta^{ab} $ knowing only $ \beta^a $ and $ \beta^b $ is as hard as taking the discrete logarithm $ a $ from $ \beta $ and $ \beta^a $ (or $ b $ knowing $ \beta $ and $ \beta^b $). Then an unauthorized third party must solve the discrete logarithm problem to the base $ \beta \in \B_n $ to determine the key. 

Both ElGamal's \cite{ElGamal} and Massey-Omura's \cite{MO} cryptosystems are essentially variants of Diffie-Hellman's key exchange system. Therefore, breaking either of the systems above requires the solution of the discrete logarithm problem for hyperoctahedral group.
\begin{defn}
	The hyperoctahedral discrete logarithm problem to the base $ \beta \in \B_n $ is the problem, given $ \pi \in \B_n $, of finding an integer $ x $ such that $ \pi = \beta^x $ if such $ x $ exists. 
\end{defn}
Many improvements are made for solving the discrete logarithm problem in finite group.
In hyperoctahedral group based cryptosystems of the sort discussed above, one does not work with the entire group $ \B_n $, but rather with cyclic subgroups : the group $ <\beta > $  in the Diffie-Hellman system and ElGamal system and the group $ < \mu > $ in the Massey-Omura system where we designate by
\begin{itemize}
\item $ <\pi >=\{\pi^i\, |\, i=1,\ldots,|\pi|\} $ the cyclic subgroup of $ \B_n $ generated by $ \pi\in\B_n $,
\item $ |\pi| $ the order of the signed permutation $ \pi $.
\end{itemize}
In other words, one works in the subgroup generated by the base of the discrete logarithm in each proposed system.
\paragraph{Choice of a suitable base}
The order of the cyclic subgroup of $ \B_n $ has an important role to avoid an easy solution of the discrete logarithm. Recall that we can also represent a signed permutation $ \pi $ in the form
\begin{equation*}
\pi = \left( \begin{array}{cccc} 1 & 2 & \dots & n\\ 
\varepsilon_1 \sigma(1) & \varepsilon_2 \sigma(2) & \dots & \varepsilon_n\sigma(n) \end{array} \right) \; 
\text{ with } \sigma\in \mathcal{S}_n \text{ and } \varepsilon_i \in \{\pm 1\}\; .
\end{equation*}
As stated in the work of Victor Reiner \cite{Reiner}, a signed permutation decomposes uniquely into a product of commuting cycles just as permutations do. 
\begin{ex}
	The disjoint cycle form of $ \left( \begin{array}{ccrcrrc} 1 & 2 & 3& 4& 5& 6 & 7\\ 
	3 & 6 & -2& 7& -5& -1 & 4 \end{array} \right) $ is  $$ \left( \begin{array}{crcr} 1 & 3& 2& 6\\ 
	3 & -2& 6& -1 \end{array} \right)\left( \begin{array}{cc} 4 & 7\\ 
	7 & 4 \end{array} \right)\left( \begin{array}{c} 5\\ 
	-5 \end{array} \right)\, . $$
\end{ex}
\begin{defn}
	The order of a signed permutation $ \pi $ is the smallest positive integer $ m $ such that $ \pi^m=\iota $ where $ \iota $ denotes the identity permutation.
\end{defn}
An $ \ell $-cycle $ C= \left( \begin{array}{ccccc} i_1 & i_2 & \dots & i_{\ell-1} & i_\ell\\ 
\varepsilon_1 i_2 & \varepsilon_2 i_3 & \dots & \varepsilon_{\ell-1} i_\ell & \varepsilon_\ell i_1 \end{array} \right) $
where $ \varepsilon_i \in \{\pm 1\} $ has order $ \ell $. The proof of the following theorem can be found in \cite{Rotman}.
\begin{thm}
	The order of a permutation written in disjoint cycle form is the least common multiple of the lengths of the cycles.	
\end{thm}
Let us now assume that we have $ \beta\in\B_n $ as base of the discrete logarithm. The order of  $ <\beta > $ can be very large for a large value of $ n $, since the disjoint cycles of $ \beta $ can be selected in such a way that their least common multiple be very large. Therefore brute force search is inefficient to solve the discrete logarithm problem to the base $ \beta \in \B_n $.
$ <\beta > $ is a cyclic group of finite order so it is commutative. Thus, in order to avoid an easy solution to the discrete logarithm problem using the techniques that apply to any finite abelian group (which take approximately $ \sqrt{p} $ operations, where $ p $ is the largest prime dividing the order of the group), it is important for the order of the commutative group $ <\beta > $ to be non smooth, that is, divisible by a large prime.
\begin{defn}
	Let $ n $ be a positive real number. we say that $ n $ is smooth if all of the prime factors of $ n $ are small.
\end{defn}
The method of Pohlig-Hellman \cite{PHellman} can efficiently computes the discrete logarithm in a group $ G $ if its order is $ B $-smooth for a reasonably small $ B $.
\begin{defn}
	Let $ B $ be a positive real number. An integer is said to be $ B $-smooth 
	if it is not divisible by any prime greater than $ B $.	
\end{defn} 
The following theorem allows to generate the base $ \beta  $ so that the order of the subgroup $ <\beta > $ of $ \B_n $ have arbitrary smoothness. 
\begin{thm}
	Let $ p $ be a prime. For a large integer $ n $, a cyclic subgroup $ G $ of $ \B_n $ which its order is $ p $-smooth and is not $ (p-1) $-smooth can be constructed.
\end{thm}
\begin{proof}
Let $ n $ be a large integer and $ p\leq n $ a prime. There exists a $ p $-cycle $ \gamma_p \in\B_n $. Let $ \gamma_p , C_1 , C_2 , \ldots , C_k\in\B_n $ be disjoint cycles of length respectively $ p , l_1 , l_2 , \ldots , l_k $ with 
$$ 1\leq l_i\leq p \text{ for } i=1,\ldots,k \text{ and } \sum_{i=1}^{k}l_i\leq n-p\; . $$
\\
Let us now consider the signed permutation $ \pi=\gamma_p  C_1  C_2  \ldots  C_k\in\B_n $ of order $$ |\pi|=lcm(p , l_1 , l_2 , \ldots , l_k)\; . $$
Let $ q $ be a prime such that $ q|lmc(p , l_1 , l_2 , \ldots , l_k) $. We show that $ q\leq p $. Let us suppose that $ q>p $. We obtain $ l_i<p<q $ which implies $ q\nmid |\pi| $ so $ q\leq p $. We have just shown that every prime $ q $ dividing the order of $ \pi $ is smaller than $ p $, that is, $ |\pi| $ is $ p $-smooth. The order $ |<\pi>|=|\pi| $ of the cyclic subgroup $ <\pi> $ of $ \B_n $ generated by $ \pi $ is  $ p $-smooth but it is not $ (p-1) $-smooth because $ p\mid |<\pi>| $ and $ p>p-1 $.
\end{proof}
This theorem provides high flexibility in selecting a subgroup of $ \B_n $ on which cryptosystem resists to attacks by Silver-Pohlig-Hellman's algorithm.

The cryptosystems that we have proposed are easy to implement by applying optimized method for exponentiation. Moreover, the multiplication on $ \B_n $ which is the composition of mappings can be performed in time $ \mathcal{O}(n) $. However, as we must work in a very large hyperoctahedral group, the need of large memory from the point of view implementation requires improvements.

\end{document}